\documentclass[conference,letterpaper]{IEEEtran}
\usepackage[utf8]{inputenc} 
\usepackage[T1]{fontenc}
\usepackage[cmex10]{amsmath} 
\usepackage{mathrsfs,amsfonts,amsthm,amssymb}
\usepackage{xcolor}
\usepackage{enumerate,lineno}
\theoremstyle{plain}
\newtheorem{theo}{Theorem}
\newtheorem{cor}[theo]{Corollary}
\newtheorem{rem}[theo]{Remark}
\newtheorem{defi}[theo]{Definition}
\newtheorem{lemma}[theo]{Lemma}

\newtheorem{prop}[theo]{Proposition}

\newtheorem{op}[theo]{Open Problem}

\newtheoremstyle{TheoremNum}
{\topsep}{\topsep}              
{\itshape}                      
{}                              
{\bfseries}                     
{.}                             
{ }                             
{\thmname{#1}\thmnote{ \bfseries #3}}
\theoremstyle{TheoremNum}
\newtheorem{theoremn}{Theorem}

\def\cM{{\mathcal{M}}}
\newcommand{\B}{\mathcal{B}}
\newcommand{\F}{\mathbb{F}}

\DeclareMathOperator{\ima}{im}
\DeclareMathOperator{\ind}{ind}

\usepackage[hyphens]{url}
\usepackage{breakurl}
\usepackage[colorlinks=true,citecolor=blue,linkcolor=blue,urlcolor=blue,bookmarks,bookmarksopen,bookmarksdepth=2,breaklinks,backref=page]{hyperref}

\usepackage{bookmark}
\bookmarksetup{
	numbered, 
	open,
}

\renewcommand*{\backref}[1]{}
\renewcommand*{\backrefalt}[4]{%
	\ifcase #1 (Not cited.)%
	\or        p.~#2.%
	\else      pp.~#2.%
	\fi}

\begin{document}
\title{When does a bent concatenation not belong to the completed Maiorana-McFarland class?}


%

 \author{%
   \IEEEauthorblockN{Sadmir Kudin\IEEEauthorrefmark{1},
                     Enes Pasalic\IEEEauthorrefmark{2},
                     Alexandr Polujan\IEEEauthorrefmark{3},
                     and Fengrong Zhang\IEEEauthorrefmark{4}}
   \IEEEauthorblockA{\IEEEauthorrefmark{1}%
                    University of Primorska, FAMNIT \& IAM, Glagoljaška 8, 6000 Koper, Slovenia,  sadmir.kudin@iam.upr.si}
    \IEEEauthorblockA{\IEEEauthorrefmark{2}%
   	University of Primorska, FAMNIT \& IAM, Glagoljaška 8, 6000 Koper, Slovenia, enes.pasalic6@gmail.com}                
   \IEEEauthorblockA{\IEEEauthorrefmark{3}%
                      Otto von Guericke University Magdeburg, Universit\"atsplatz 2, 39106 Magdeburg, Germany, alexandr.polujan@gmail.com}
   \IEEEauthorblockA{\IEEEauthorrefmark{4}%
                     School of Cyber Engineering, Xidian University, Xi'an 710071, P.R. China, zhfl203@163.com}
 }

\maketitle
\begin{abstract}
	Every Boolean bent function $f$  can be written either as a concatenation $f=f_1||f_2$ of two complementary semi-bent functions $f_1,f_2$;  or as a concatenation $f=f_1||f_2||f_3||f_4$ of four Boolean functions $f_1,f_2,f_3,f_4$, all of which are simultaneously bent, semi-bent, or 5-valued spectra-functions. In this context, it is essential to ask: When does a bent concatenation $f$ (not) belong to the completed Maiorana-McFarland class $\mathcal{M}^\#$?  In this article, we answer this question completely by providing a full characterization of the structure of $\mathcal{M}$-subspaces for the concatenation of the form $f=f_1||f_2$ and $f=f_1||f_2||f_3||f_4$, which allows us to specify the necessary and sufficient conditions so that $f$ is outside $\mathcal{M}^\#$. Based on these conditions, we propose several explicit design methods of specifying bent functions outside $\mathcal{M}^\#$ in the special case when $f=g||h||g||(h+1)$, where $g$ and $h$ are bent functions.
\end{abstract}

\section{Preliminaries}\label{sec:pre}
Let $\mathbb{F}_2^n$ be the vector space of all $n$-tuples $x=(x_1,\ldots,x_n)$, where $x_i \in \mathbb{F}_2$. For $x=(x_1,\ldots,x_n)$ and $y=(y_1,\ldots,y_n)$ in $\mathbb{F}^n_2$, the usual scalar product over $\mathbb{F}_2$ is defined as $x\cdot y=x_1 y_1+\cdots+ x_n y_n.$ By $0_n$ we denote the all-zero vector of $\mathbb{F}^n_2.$ Every Boolean function $f\colon\mathbb{F}^n_2 \rightarrow \mathbb{F}_2$ can be uniquely represented by its associated algebraic normal form (ANF) in the form $f(x_1,\ldots,x_n)=\sum_{u\in \mathbb{F}^n_2}{\lambda_u}{(\prod_{i=1}^n{x_i}^{u_i})}$,  where $x_i, \lambda_u \in \mathbb{F}_2$ and $u=(u_1, \ldots,u_n)\in \mathbb{F}^n_2$.
	The algebraic degree of $f$, denoted by $\deg(f)$, is equal to the maximum Hamming weight of $u \in \F_2^n$ for which $\lambda_u \neq 0$.
	
	The \textit{first-order derivative} of a function $f$ in the direction $a \in \F_2^n$ is given by $D_{a}f(x)=f(x)+ f(x+ a)$. Derivatives of higher orders are defined recursively, i.e., the \textit{$k$-th order derivative} of a function $f\in\mathcal{B}_n$ is defined by $D_Vf(x)=D_{a_k}D_{a_{k-1}}\ldots D_{a_1}f(x)=D_{a_k}(D_{a_{k-1}}\ldots D_{a_1}f)(x)$, where $V=\langle a_1,\ldots,a_k \rangle$ is a vector subspace of $\F_2^n$ spanned by elements $a_1,\ldots,a_k\in\F_2^n$. Note that if $a_1,\ldots,a_k\in\F_2^n$ are linearly dependent, then $D_{a_k}D_{a_{k-1}}\ldots D_{a_1}f=0$. The \emph{Walsh-Hadamard transform} of $f\in\mathcal{B}_n$ at any point $\omega\in\mathbb{F}^n_2$ is defined $
	W_{f}(\omega)=\sum_{x\in \mathbb{F}_2^n}(-1)^{f(x)\oplus \omega\cdot x}$.  A function $f\in\mathcal{B}_n,$ for even $n$, is called \textit{bent} if $|W_f(u)|=2^{\frac{n}{2}}$, for all $u\in\F_2^n$. Its unique {\it dual} function $f^*$ is defined as $W_f(u)=2^{\frac{n}{2}}(-1)^{f^*(u)}$, which is also bent. Two Boolean functions $f,f'\in\mathcal{B}_n$ are called \textit{extended-affine equivalent}, if there exists an affine permutation $A$ of $\F_2^n$ and affine function $l\in\mathcal{B}_n$, such that $f\circ A + l= f'$. It is well known, that extended-affine (EA) equivalence preserves the bent property. 
	
	The \textit{completed Maiorana-McFarland class $\cM^\#$} \cite{MM73} is the set of $n$-variable ($n=2m$) Boolean bent functions, which are EA-equivalent to the functions of the form
\begin{equation}\label{eq:MMdefinition}
	f(x,y)=x \cdot \pi(y)+ g(y), \mbox{ for all } x, y\in\F_2^m,
\end{equation}
where $\pi$ is a permutation on $\F_2^m$, and $g$ is an arbitrary Boolean function on
$\F_2^m$.  It is well-known from Dillon's thesis~\cite{Dillon} that a bent function $f\in\mathcal{B}_n$ belongs to  $\mathcal{M}^\#$ iff there exists a vector space $V$ of dimension $m$, such that $D_a D_b f=0$ for all $a,b\in V$. This characterization motivates the following definition:

\begin{defi}\cite{Polujan2020}
	Let $f\in\mathcal{B}_n$ be a Boolean function. We call a vector subspace $V$ of $\F_2^n$  an $\mathcal{M}$-subspace of $f$, if we have that 	$ D_{a}D_{b}f=0$, for any $ a,  b \in V$. 
\end{defi} 

Further, we will investigate $\mathcal{M}$-subspaces of the Boolean functions of the form $f=f_1||f_2$ or $f=f_1||f_2||f_3||f_4$, which are defined as follows. We define the concatenation $f_1 \vert \vert f_2: \F_2^{n+1} \to \F_2$ of the two functions as:
\begin{equation}\label{fform}
	\begin{split}
		&f_1 \vert \vert f_2( z,z_{n+1}) = f_1(z)+z_{n+1} (f_1(z)+f_2(z)),\\& \text{ for all } z \in \F_2^n, \; z_{n+1}\in \F_2,
	\end{split}
\end{equation}
that is, $f_1 \vert \vert f_2( z,0)=f_1(z)$, and $f_1 \vert \vert f_2( z,1)=f_2(z)$.

For $i=1, \dots, 4$, let $f_i\in\mathcal{B}_n$. The formula for the concatenation $f=f_1 \vert \vert f_2 \vert \vert f_3 \vert \vert f_4\in\mathcal{B}_{n+2}$ of the four functions is given by:
\begin{equation}\label{fform4}
	\begin{split}
		f( z,z_{n+1},z_{n+2}) =& f_1(z)+ z_{n+1}z_{n+2} (f_1+f_2+f_3+f_4)(z)\\ +& z_{n+1} (f_1+f_2)(z)+z_{n+2} (f_1+f_3)(z),
	\end{split}
\end{equation}
for all $z \in \F_2^n$ and $z_{n+1},z_{n+2} \in \F_2$, that is, $f( z,0,0)=f_1(z)$, $f( z,1,0)=f_2(z)$, $f( z,0,1)=f_3(z)$ and $f( z,1,1)=f_4(z)$. Throughout this article, we will call bent functions of the form~\eqref{fform} and \eqref{fform4} \textit{bent concatenations}.

The main aim of this article is to develop further a theory of $\mathcal{M}$-subspaces for bent concatenations initially analyzed in \cite{Polujan2020} and recently considered in \cite{PPKZ2023}.  For a more detailed treatment of bent functions we refer to \cite{CarlMes2016,Mesnager}, and for their designs outside $\cM^\#$ to \cite{Bent_Decomp2022,PPKZ_BFA23_CCDS}.  The rest of the paper is organized in the following way.
In Sections~\ref{sec: concat 2} and~\ref{sec: concat 4}, we provide a full characterization of the structure of $\mathcal{M}$-subspaces for the concatenation of the form $f=f_1||f_2$ and $f=f_1||f_2||f_3||f_4$, respectively. Consequently, we specify the necessary and sufficient conditions so that $f$ is outside $\cM^\#$. Based on these conditions, we propose in Section~\ref{sec:design} several explicit design methods of specifying bent functions outside $\cM^\#$ in the special case when $f=g||h||g||(h+1)$.

\section{Concatenation of two Functions}\label{sec: concat 2}

Let $a,b\in\F_2^n$. From Eq.~\eqref{fform}, we deduce that the second-order derivative of the concatenation $f=f_1 \vert \vert f_2: \F_2^{n+1} \to \F_2$, with respect to $(a,0)$ and $(b,0)$ has the following form
\begin{equation}\label{eq: secder00}
	D_{(a,0)}D_{(b,0)}f=D_{(a,0)}D_{(b,0)}f_1 \vert \vert f_2= D_aD_bf_1 \vert \vert D_aD_bf_2.
\end{equation}
Similarly, from Eq.~\eqref{fform}, the second-order derivative of $f=f_1 \vert \vert f_2$ w.r.t. $(a,0)$ and $(b,1)$, at the point $(z,z_{n+1}) \in \F_2^{n+1}$, can be computed as
\begin{equation}\label{eq: secder01}
	\begin{split}
		&D_{(a,0)}D_{(b,1)}f= D_{(b,1)}(D_af_1 \vert \vert D_af_2)= g_1 \vert \vert g_2, \text{ where } \\
		&g_1(z)=D_af_1(z)+D_af_2(z+b) \text{ and }\\ 
		&g_2(z)=D_af_2(z)+D_af_1(z+b), \text{ for all } z \in \F_2^n.
	\end{split}
\end{equation}
Since $D_{(a,a_{n+1})}D_{(b,b_{n+1})}f=D_{(b,b_{n+1})}D_{(a,a_{n+1})}f=D_{(a+b,a_{n+1}+b_{n+1})}D_{(b,b_{n+1})}f$, for all $a,b \in \F_2^n$ and $a_{n+1}, b_{n+1} \in \F_2$, the rest of the cases can also be computed with \eqref{eq: secder00} and \eqref{eq: secder01}. Using these expressions, we relate $\mathcal{M}$-subspaces of $f$ to $\mathcal{M}$-subspaces of $f_1$ and $f_2$ as follows:

\begin{theo} \label{th: Msubspacesoff1f2} Let $f_1,f_2 \in \mathcal{B}_{n}$ and let $k \in \{ 1, \dots ,n \}$. The function $f=f_1 \vert \vert f_2 \in \mathcal{B}_{n+1}$ has no $(k+1)$-dimensional $\cM$-subspaces if and only if the following conditions hold:
	\begin{enumerate}[a)]
		\item The functions $f_1$ and $f_2$ do not share a common $(k+1)$-dimensional $\cM$-subspace;
		\item For every vector $u \in \F_2^n$ and every $k$-dimensional $\cM$-subspace $V \subset \F_2^n$ of both $f_1$ and $f_2$, there is $a \in V$ such that
		\begin{equation}
			\label{eq:conditoutsideM} D_af_1(z)+D_af_2(z+u) \neq 0, \text{ for some } z \in \F_2^n.
		\end{equation}
	\end{enumerate}
\end{theo}
\begin{proof}
	(Sketch) Assume that $W$ is an $\cM$-subspace of $f$, with $\dim(W)=k+1$. Consider the projection $P: W \to \F_2$ given by $P(z,z_{n+1})=z_{n+1}$, for all $(z,z_{n+1}) \in W$, where $z \in \F_2^n$ and $z_{n+1} \in \F_2$. Then, $\dim (\ker(P)) \geq k$ (by rank-nullity theorem). If $\dim (\ker(P)) =k+1$, then Eq.~\eqref{eq: secder00} implies that $f_1$ and $f_2$ share a common $(k+1)$-dimensional $\cM$-subspace. Similarly, when $\dim (\ker(P)) =k$, define $V$ through $\lbrace (v,0) \colon v\in V \rbrace = \ker(P)$. Then, taking  $u \in \F_2^n$ be such that $(u,1) \in W \setminus \ker(P)$, by Eqs.~\eqref{eq: secder00} and \eqref{eq: secder01} one deduces Eq. \eqref{eq:conditoutsideM}. 
	In the other direction, it can be shown that assuming that $f_1$ and $f_2$ do not share a common $(k+1)$-dimensional $\cM$-subspace leads to a contradiction.
\end{proof}

Using the fact that a bent function $f \in \mathcal{B}_{t}$ is in the $\cM^{\#}$ class if and only if it has a $t/2$-dimensional $\cM$-subspace, from Theorem \ref{th: Msubspacesoff1f2} we deduce the following result.

\begin{cor} \label{cor:completecharact} Let $f_1,f_2 \in \mathcal{B}_{n}$, $n=2k+1$, be Boolean functions such that $f=f_1 \vert \vert f_2 \in \mathcal{B}_{n+1}$ is a bent function. Then, the function $f$ is outside the $\cM^{\#}$ class if and only if the following conditions hold:
	\begin{enumerate}
		\item The functions $f_1$ and $f_2$ do not share a common $(k+1)$-dimensional $\cM$-subspace;
		\item For every vector $u \in \F_2^n$ and every $k$-dimensional $\cM$-subspace $V \subset \F_2^n$ of both $f_1$ and $f_2$, there is $a \in V$ such that
		$D_af_1(z)+D_af_2(z+u) \neq 0, \text{ for some } z \in \F_2^n$.
	\end{enumerate}
\end{cor}

It is well-known that in the above concatenation $f=f_1||f_2$, the function $f$ is bent if and only if $f_1$ and $f_2$ are disjoint spectra semi-bent functions; see~\cite[Theorem 6]{Zheng2001}. In particular, when $f_i\colon\F_2^{2k+1} \rightarrow \F_2$ are represented in the form $f_i(x,y)=x \cdot \phi_i(y) + h_i(y)$, for $x \in \F_2^{k+1}$, $y \in \F_2^k$, where $\phi\colon \F_2^{k} \rightarrow \F_2^{k+1}$ and $h_i\colon\F_2^k \rightarrow \F_2$, then the properties of $\phi_i$ are essential in defining disjoint spectra semi-bent functions $f_1$ and $f_2$.

\begin{theo} \label{th:bentcond} Let  $f_1$ and $f_2$ defined as $f_i(x,y)=x \cdot \pi_i(y) + h_i(y)$, with $x \in \F_2^{k+1}$ and $y \in \F_2^{k}$ and $h_i$ are arbitrary Boolean  functions on $\F_2^{k}$. Then, the concatenation $f=f_1||f_2$ is a bent   function on $\F_2^{2k+2}$ if and only if $\ima(\pi_1) \cap \ima(\pi_2)=\varnothing$ and $\pi_i$ are  injective mappings.
	
\end{theo}
\begin{proof} 
	Notice that $f=f_1||f_2 \colon \F_2^{k+1} \times \F_2^{k+1} \to \F_2$ is the function defined by $f(x,y)=x \cdot \pi (y,y_{k+1})+h(y,y_{k+1})$,  for all $x \in \F_2^{k+1}, \; y \in \F_2^k$ and $y_{n+1} \in \F_2,$ where $\pi$ is defined by $\pi(y,0)=\pi_1(y)$ and $\pi(y,1)=\pi_2(y)$, and similarly $h(y,0)=h_1(y)$ and $h(y,1)=h_2(y)$, for all $y \in \F_2^k$.
	We know that $f$ is bent if and only if $\pi$ is a permutation, and $\pi$ is a permutation if and only if $\ima(\pi_1) \cap \ima(\pi_2)=\varnothing$ and $\pi_1$ and $\pi_2$ are injective mappings.
\end{proof}
However,  it turns out that $f=f_1||f_2 \in \cM^\#$ since $f_1$ and $f_2$ share an $\mathcal{M}$-subspace of maximal dimension.
\begin{rem}
Any construction method employing the functions $f_i(x,y)=x \cdot \phi_i(y) + h_i(y)$, where $x \in \F_2^{k+1}$ and $y \in \F_2^k$ (consequently $\phi_i : \F_2^k \rightarrow \F_2^{k+1}$), will only provide a function $f$ which belongs to $\cM^\#$. This is due to Corollary~\ref{cor:completecharact} and the fact that $\F_2^{k+1} \times \{0_k\}$ is a canonical $\mathcal{M}$-subspace of dimension $k+1$ which is shared by $f_1$ and $f_2$.
\end{rem}

\section{Concatenation of four Functions}\label{sec: concat 4}

Similarly as in the case of two functions concatenation, we derive the following formulas for the second-order derivatives of $f=f_1 \vert \vert f_2 \vert \vert f_3 \vert \vert f_4$ (where $f_i$ are  suitable bent, semi-bent or five-valued spectra functions) if $f$ is bent \cite{Decom}). For a function $h: \F_2^m \to \F_2$ and $r \in \F_2^m$ by $h^r$, we denote the translation of $h$ by $r$, that is $h^r(x)=h(x+r)$, for all $x \in \F_2^m$. In the following formulas, $a$ and $b$ are two arbitrary elements from $\F_2^n$, not necessarily different.
\begin{equation}\label{eq: secder4_0000}
	\begin{split}
		&D_{(a,0,0)}D_{(b,0,0)}f=D_{(a,0,0)}D_{(b,0,0)}(f_1 \vert \vert f_2 \vert \vert f_3 \vert \vert f_4)\\ & =D_aD_bf_1 \vert \vert D_aD_bf_2 \vert \vert D_aD_bf_3 \vert \vert D_aD_bf_4
	\end{split}	
\end{equation}
\begin{equation}\label{eq: secder4_1000}
	\begin{split}
		&D_{(a,1,0)}D_{(b,0,0)}f=(D_bf_1+D_bf_2^a) \vert \vert \\& (D_bf_1+D_bf_2^a)^a \vert \vert  (D_bf_3+D_bf_4^a) \vert \vert  (D_bf_3+D_bf_4^a)^a
	\end{split}
\end{equation}
\begin{equation}\label{eq: secder4_0100}
	\begin{split}
		&D_{(a,0,1)}D_{(b,0,0)}f=(D_bf_1+D_bf_3^a) \vert \vert \\ & (D_bf_2+D_bf_4^a) \vert \vert (D_bf_1+D_bf_3^a)^a \vert \vert (D_bf_2+D_bf_4^a)^a
	\end{split}
\end{equation}
\begin{equation}\label{eq: secder4_1100}
	\begin{split}
		&D_{(a,1,1)}D_{(b,0,0)}f=(D_bf_1+D_bf_4^a) \vert \vert \\ & (D_bf_2+D_bf_3^a) \vert \vert  (D_bf_2+D_bf_3^a)^a \vert \vert (D_bf_1+D_bf_4^a)^a
	\end{split}
\end{equation}
\begin{equation}\label{eq: secder4_0110}
	\begin{split}
		&D_{(a,0,1)}D_{(b,1,0)}f= (f_1 +f_2^b + f_3^a +f_4^{a+b}) \vert \vert \\& (f_1 +f_2^b + f_3^a +f_4^{a+b})^b \vert \vert (f_1 +f_2^b + f_3^a +f_4^{a+b})^a \vert \vert \\& (f_1 +f_2^b + f_3^a +f_4^{a+b})^{a+b}.
	\end{split}
\end{equation}

Compared to Proposition V.2 in \cite{PPKZ2023}, the result below gives the most general structure of $\mathcal{M}$-subspaces of varying dimension for a 4-concatenation of not necessarily bent functions. 
\begin{theo}\label{th: formofMsubspaces4concatenation}
	Let $f=f_1 \vert \vert f_2 \vert \vert f_3 \vert \vert f_4 \colon \F_2^{n+2} \to \F_2$ be the concatenation of arbitrary Boolean functions $f_1, \ldots ,f_4 \in \mathcal{B}_n$ and let $W$ be a $(k+2)$-dimensional subspace of $\F_2^{n+2}$, $k \in \{0, \ldots,n \}$. Then, $W$ is an $\cM$-subspace of $f$ if and only if $W$ has one of the following forms:
	\begin{enumerate}[a)]
		\item \label{item a, th: form} $W= V \times \{(0,0) \}$, where $V \subset \F_2^n$ is a common $(k+2)$-dimensional $\cM$-subspace of $f_1, \ldots ,f_4$.
		\item $W= \langle V \times \{(0,0)\}, (a,1,0) \rangle$, where $V$ is a common $(k+1)$-dimensional $\cM$-subspace of $f_1, \ldots ,f_4$, and $a \in \F_2^n$ is such that 
		$$D_vf_1+D_vf_2^a= D_vf_3+D_vf_4^a=0, \text{ for all } v \in V.$$
		\item $W= \langle V \times \{(0,0)\}, (a,0,1) \rangle$, where $V$ is a common $(k+1)$-dimensional $\cM$-subspace of $f_1, \ldots ,f_4$, and $a \in \F_2^n$ is such that 
		$$D_vf_1+D_vf_3^a= D_vf_2+D_vf_4^a=0, \text{ for all } v \in V.$$
		\item $W= \langle V \times \{(0,0)\}, (a,1,1) \rangle$, where $V$ is a common $(k+1)$-dimensional $\cM$-subspace of $f_1, \ldots ,f_4$, and $a \in \F_2^n$ is such that 
		$$D_vf_1+D_vf_4^a= D_vf_2+D_vf_3^a=0, \text{ for all } v \in V.$$
		\item \label{item e, th: form}$W= \langle V \times \{(0,0)\}, (a,0,1), (b,1,0) \rangle$, where $V$ is a common $k$-dimensional $\cM$-subspace of $f_1, \ldots ,f_4$, and $a,b \in \F_2^n$ are such that 
		$D_vf_1+D_vf_3^a= D_vf_2+D_vf_4^a=D_vf_1+D_vf_2^b= D_vf_3+D_vf_4^b=0, \text{ for all } v \in V$,  and
		$f_1(x)+f_2(x+b)+f_3(x+a)+f_4(x+a+b)=0$, for all  $x \in \F_2^n$.
	\end{enumerate}
\end{theo}
\begin{proof}(Sketch) Assume first that $W$ is an $\cM$-subspace of $f$. Let $P \colon W \to \F_2^2$ be the projection on the last two coordinates, i.e., $P((w_1, \dots ,w_{n+1},w_{n+2}))=(w_{n+1},w_{n+2})$, for all $(w_1, \dots ,w_{n+1},w_{n+2}) \in W$.
	There are $5$ subspaces of $\F_2^2$, and depending on which subspace $\ima(P)$ is equal to, we obtain the five corresponding forms \ref{item a, th: form}) - \ref{item e, th: form}) of the subspace $W$. The proof follows by applying Eqs. \eqref{eq: secder4_0000} - \eqref{eq: secder4_0110}. The other direction is proved similarly.
\end{proof}

\begin{rem}\label{rem:PropositionV2}
	Proposition V.2 in \cite{PPKZ2023}	specifies the structure of $\mathcal{M}$-subspaces of maximal dimension $m+1$ for $f=f_1||f_2||f_3||f_4$, where both $f$ and $f_i \in \B_{2m}$ are bent and additionally at least one $f_i$ admits the canonical $\mathcal{M}$-subspace $U=\F_2^m \times  \{0_m\}$. Thus, it is a special case of Theorem \ref{th: formofMsubspaces4concatenation}.
\end{rem}

From Theorem \ref{th: formofMsubspaces4concatenation}, we obtain the following full characterization of the class inclusion of $f=f_1 \vert \vert f_2 \vert \vert f_3 \vert \vert f_4$ in the $\cM^{\#}$ class in terms of properties of $f_1, \ldots, f_4$.

\begin{cor}\label{cor: 4concatenationoutsideMM}
	Let $f=f_1 \vert \vert f_2 \vert \vert f_3 \vert \vert f_4 \colon \F_2^{n+2} \to \F_2$ be the concatenation of $f_1, \ldots ,f_4 \in \mathcal{B}_n$ and assume that $f$ is bent; thus $f_i$ are bent, semi-bent or five-valued spectra functions. Then, $f$ is outside of the $\cM^{\#}$ class if and only if the following conditions hold:
	\begin{enumerate}[a)]
		\item The functions $f_1, \dots ,f_4$ do not share a common $(n/2+1)$-dimensional $\cM$-subspace; 
		\item There are no common $(n/2)$-dimensional $\cM$-subspaces $V \subset \F_2^n$ of $f_1, \dots ,f_4$ such that there is an element $a\in \F_2^n$ for which
		\begin{equation}
			\begin{gathered}
				D_vf_1+D_vf_2^a= D_vf_3+D_vf_4^a=0, \text{ for all } v \in V, \text{ or} \\
				D_vf_1+D_vf_3^a= D_vf_2+D_vf_4^a=0, \text{ for all } v \in V, \text{ or} \\
				D_vf_1+D_vf_4^a= D_vf_2+D_vf_3^a=0, \text{ for all } v \in V.
			\end{gathered}
		\end{equation}
		\item There are no common $(n/2-1)$-dimensional $\cM$-subspaces $V \subset \F_2^n$ of $f_1, \dots ,f_4$ such that there are elements $a,b\in \F_2^n$ (not necessarily different), for which
		\begin{equation}
			\begin{split}
				&D_vf_1+D_vf_3^a= D_vf_2+D_vf_4^a=D_vf_1+D_vf_2^b\\ & =D_vf_3+D_vf_4^b=0, \text{ for all } v \in V, \text{ and} \\
				&f_1(x)+f_2(x+b)+f_3(x+a)\\&+f_4(x+a+b)=0, \text{ for all } x \in \F_2^n.
			\end{split}
		\end{equation}
	\end{enumerate}
\end{cor}
\begin{proof}
	The result follows directly from Theorem \ref{th: formofMsubspaces4concatenation}, by setting $k+2=n/2+1$, and the fact that a bent function $f \in \mathcal{B}_{n+2}$ is in the $\cM^{\#}$ class if and only if it has an $(n/2+1)$-dimensional $\cM$-subspace. 
\end{proof}
Notice that when $f_i$ are bent in Corollary \ref{cor: 4concatenationoutsideMM}, then  the item $a)$ is automatically satisfied since none of the functions $f_i$ admits an $\mathcal{M}$-subspace of dimension $n/2+1$. The condition in $b)$ was recently deduced in \cite[Corollary V.11]{PPKZ2023} for a special case  when $f_i$ are bent functions on $\F_2^n$ that share an $\mathcal{M}$-subspace of maximal dimension $n/2$. 
\begin{op}\label{op:sufficiency}
	Is the condition $c)$ in Corollary \ref{cor: 4concatenationoutsideMM} independent of conditions $a),b)$? Particularly, the existence of bent functions $f=f_1 \vert \vert f_2 \vert \vert f_3 \vert \vert f_4$ on $\F_2^{n+2}$ in $\cM^\#$, where all $f_i\in\mathcal{B}_n$ are bent and outside $\cM^\#$, is hard to establish.
\end{op}
Notice that, when $f=f_1||f_1||f_1||f_1+1$ so that $f(x,y_1,y_2)=f_1(x)+y_1y_2$, where $f_1$ is a bent function on $\F_2^n$, it was deduced \cite{ZPBBInfComp} that  $f$ is outside $\cM^\#$ if and only if $f_1$ is outside $\cM^\#$. This result also follows from Theorem~\ref{th: insideMMgh} below, as we show in the next section.

\section{An Application: Designing bent functions outside $\cM^\#$ of the form $g||h||g||(h+1)$}\label{sec:design}
The concatenation $f=g||h||g||h+1$  (where $g$ and $h$ are bent) is interesting in terms of the class inclusion, as the dual bent condition is automatically satisfied. Recall that when $f_i$ are all bent, then $f=f_1||f_2||f_3||f_4$ is bent if and only if $f_1^* + f_2^* + f_3^* +f_4^*=1$; see \cite{SHCF}.
The analysis of structural properties of $\cM$-subspaces presented in the previous section turns out to be useful when considering certain special cases of bent 4-concatenation. 
\subsection{The necessary and sufficient condition for $f=g||h||g||(h+1)$ to be outside $\mathcal{M}^\#$}
\begin{theo}\label{th: insideMMgh} Let $h$ and $g$ be two arbitrary bent functions in $\mathcal{B}_n$. Then, the function $f=f_1 \vert \vert f_2 \vert \vert f_3 \vert \vert f_4 \colon \F_2^{n+2} \to \F_2$, where $f_1=f_3=g$ and $f_2=f_4+1=h$ is a bent function in the $\cM^{\#}$ class if and only if the functions $g$ and $h$ have a common $(n/2)$-dimensional $\cM$-subspace, thus $g, h \in \cM^\#$.
\end{theo}
\begin{proof}
	We compute $f_1^*+f_2^*+f_3^*+f_4^*=g^*+h^*+g^*+h^*+1=1$, hence $f$ is a bent function. Let $V \subset \F_2^n$ be a common $(n/2)$-dimensional $\cM$-subspace of $g$ and $h$. Then, $V$ is also a common $(n/2)$-dimensional $\cM$-subspace of $f_1, \ldots ,f_4$ and $D_vf_1+D_vf_3=D_vg+D_vg=0$, $D_vf_2+D_vf_4=D_vh+D_vh=0$,
	for all $v \in V$. Setting $a=0_n$ in the item $b)$ of Corollary \ref{cor: 4concatenationoutsideMM}, we deduce that $f$ is a bent function in $\cM^{\#}$.
	
	Assume now that $g$ and $h$ do not have a common $(n/2)$-dimensional $\cM$-subspace, and that $f \in \cM^{\#}$. 
	Then, the cases $a)$ and $b)$ in Corollary \ref{cor: 4concatenationoutsideMM} hold, hence it has to be the case $c)$ that fails. That is, there is 
	a common $(n/2-1)$-dimensional $\cM$-subspace $V \subset \F_2^n$ of $f_1, \dots ,f_4$, (i.e. of $g$ and $h$)
	such that there are elements $a,b\in \F_2^n$ (not necessarily different), for which
	\begin{equation*}
		\begin{split}
			&D_vf_1+D_vf_3^a= D_vf_2+D_vf_4^a=D_vf_1+D_vf_2^b=\\ &D_vf_3+D_vf_4^b=0, \text{ for all } v \in V, \text{ and} \\
			&f_1(x)+f_2(x+b)+f_3(x+a)+f_4(x+a+b)=0,\\& \text{ for all } x \in \F_2^n.
		\end{split}
	\end{equation*}
	From $D_vf_1+D_vf_3^a=0$, we get $D_vg+D_vg^a=D_aD_vg=0$, for all $v \in V$. Similarly,  $D_vf_2+D_vf_4^a=0$ implies $D_vh+D_vh^a=D_aD_vh=0$, for all $v \in V$. This implies that $a$ has to be in $V$, otherwise $\langle V ,a \rangle$ would be a common $(n/2)$-dimensional $\cM$-subspace of $g$ and $h$.
	Setting $v=a$ in $D_vf_1+D_vf_2^b=0$, we get 
	\begin{equation}\label{eq: zerosuumofgh}
		\begin{split}
			&g(x)+g(x+a)+h(x+b)+h(x+a+b)=0,\\
			&\text{ for all } x \in \F_2^n.
		\end{split}
	\end{equation} 
		On the other hand, from $f_1(x)+f_2(x+b)+f_3(x+a)+f_4(x+a+b)=0$ we have $g(x)+h(x+b)+g(x+a)+h(x+a+b)+1=0$,
		that is $g(x)+g(x+a)+h(x+b)+h(x+a+b)=1$, for all $x \in \F_2^n$. However, this is in contradiction with Eq.\ \eqref{eq: zerosuumofgh}. We conclude that $f$ is a bent function outside the $\cM^{\#}$ class.
\end{proof}

	\begin{rem} 
		Notice that Theorem \ref{th: insideMMgh} answers negatively Open Problem \ref{op:sufficiency} when a bent function $f \in \B_{n+2}$ is represented as  $f=g||h||g||h+1$.
	\end{rem}
However, Theorem \ref{th: insideMMgh} provides a very flexible method of constructing bent functions outside $\cM^\#$ for $n \geq 10$. 
\begin{cor}\label{cor: explicitdesign}
	Let $g \in \mathcal{B}_n$ be any bent function outside $\cM^\#$, with $n \geq 8$, and $h$ be any bent function on $\F_2^n$. Then, the bent function $f \in \mathcal{B}_{n+2}$ defined as $f=g||h||g||h+1$ is outside the $\cM^{\#}$ class.
\end{cor}
\begin{proof}
	By Theorem \ref{th: insideMMgh}, $f \in \cM^\#$ if and only if $g$ and $h$ share a common $(n/2)$-dimensional $\cM$-subspace. But since $g$ is outside $\cM^{\#}$ it does not admit any $(n/2)$-dimensional $\cM$-subspace, and therefore it cannot share with $h$ regardless of $h$ belongs to $\cM^{\#}$ or not. Thus, $f \in \B_{n+2}$ is outside $\cM^{\#}$.
\end{proof}
Another important consequence of Theorem \ref{th: insideMMgh} is the following result which also sheds  more light on the existence of bent functions outside $\cM^{\#}$, for the special case when $n=8$. 
\begin{cor}\label{cor: withoutaffinederivatives1}
	Let $g \in \mathcal{B}_n$ be an arbitrary  bent function $n \geq 6$. Then, there exists a bent function $f \in \mathcal{B}_{n+2}$ outside the $\cM^{\#}$ class such that $g(x)=f(x,0,0)$, for all $x \in \F_2^n$.
\end{cor}
\begin{proof}
	Let $h$ be a bent function in $n$ variables with a unique $(n/2)$-dimensional $\cM$-subspace $V$; see  \cite{PPKZ2023} for their  existence.
	Since $g$ is bent, thus not affine, there exist are two elements $a,b \in \F_2^n$ such that $D_aD_bg \neq 0$. Let $A$ be any affine permutation of $\F_2^n$ such that $A^{-1}( \{a,b\}) \subset V$. Define $h'=h \circ A$. Then, by construction $g$ and $h'$ do not share an $(n/2)$-dimensional $\cM$-subspace. Therefore, by Theorem \ref{th: insideMMgh}, the function 
	$
	f=g \vert \vert h' \vert \vert g \vert \vert (h'+1)
	$
	is a bent function outside the $\cM^{\#}$ class, and the result follows.
\end{proof}
 Note that certain  design methods of constructing 8-variable bent functions outside $\cM^{\#}$ using bent functions $f_1, \ldots, f_4 \in \cM^\#$ were considered in \cite{PPKZ2023}, but Corollary \ref{cor: withoutaffinederivatives1} confirms this fact theoretically and thus excludes the case that bent functions outside $\cM^\#$ originate from the 4-concatenation of semi-bent or five-valued spectra functions only. Moreover, it is always possible to find more than one permutation $A$ (from the proof of Corollary~\ref{cor: withoutaffinederivatives1}). It means that for $n \geq 6$, the number of bent functions outside $\mathcal{M}^\#$ in $n+2$ variables is always strictly greater than the number of all bent functions in $n$ variables.


	\begin{theo}\label{th: r-ind(f)}
		Let $n,k$ be two integers such that $k<n/2-1$.  Let  $g,h$ be two bent functions in $\mathcal{B}_n$ whose $\mathcal{M}$-subspaces of maximal dimension $k$ are mutually non-intersecting. 
		Assume that for any subspace $\Lambda \subset \F_2^n$ with $\dim(\Lambda)=k-1$, there exists $a\in \Lambda$ such that 
		$D_ag\neq D_ah.$
		Then,  $f=f_1 \vert \vert f_2 \vert \vert f_3 \vert \vert f_4 \colon \F_2^{n+2} \to \F_2$, where $f_1=f_3=g$ and $f_2=f_4+1=h$, is a bent function whose $\mathcal{M}$-subspaces have dimension $<k+1$.
	\end{theo}
	\begin{proof} (Sketch)  By assumption,  we have that $W= \langle V \times \{(0,0)\}, (a,i_1,i_2) \rangle$ is not an $\cM$-subspace of $f$, where $V$ is a  $k$-dimensional $\cM$-subspace of $g$ (resp. $h$), $(i_1,i_2)\in \F_2^2$. 
		Thus, let  $\Delta$ be a common $(k-1)$-dimensional $\cM$-subspace of $g$ and $h$.  Set
		$W= \langle \Delta \times \{(0,0)\}, (a,i_1,i_2),  (b,i_3,i_4) \rangle$, where $(i_1,i_2), (i_3,i_4)\in \F_2^2$ and $(i_1,i_2)\neq (i_3,i_4)$.
		Then, there are two cases to be considered, namely 1) $a\notin \Delta$ or $b\notin \Delta$ and 2) $a,b\in \Delta$. It can be shown that the vector space $W$, with $\dim(W)=k+1$, is not an $\cM$-subspace of $f$.
		The result follows then from Theorem \ref{th: formofMsubspaces4concatenation}.
	\end{proof}
%
%
	\begin{cor}\label{cor: r-ind(f)} Let $n$ be an even integer.  Let $\pi$ be a permutation such that   $g=x \cdot \pi(y)$ has only one  $(n/2)$-dimensional $\cM$-subspace.  
		Let $A$ be an invertible matrix on $\F_2^n$ such that $I+A$ is also an invertible matrix on $\F_2^n$.  Let $h= g\circ A$.
		Then, the function $f=f_1 \vert \vert f_2 \vert \vert f_3 \vert \vert f_4 \colon \F_2^{n+2} \to \F_2$, where $f_1=f_3=g$ and $f_2=f_4+1=h$, is a bent function outside  $\cM^{\#}$.
	\end{cor}
	\begin{proof}
		Since $g$ has only one  $(n/2)$-dimensional $\cM$-subspace, we have that $g$ and $h$ have no common 	$(n/2)$-dimensional $\cM$-subspace. Since  $I+A$ is also an invertible matrix on $\F_2^n$, we have $g+h$ is also a bent function, that is,  for any nonzero  vector $a\in \F_2^n$ we have
		$D_ag\neq D_ah.$  From Theorem \ref{th: r-ind(f)}, we have the maximal dimension of $\mathcal{M}$-subspaces of $f$ is $<n/2+1$, thus $f \not \in \cM^{\#}$.
	\end{proof}

\subsection{A special case of relating $g$ and $h$ in a linear manner}
The following result, obtained in \cite{Korsakova2013}, provides two secondary constructions of bent functions in $n+2$ variables from bent functions in $n$ variables.  Notice that a version of the result is also stated as Theorem 45 in \cite{Tokareva}. 

\begin{theo}\label{th: Korsakova} \cite{Korsakova2013}
	Let $g$ be a bent function in $n$ variables. Then, the functions $f$ and $f'$ in $(n+2)$-variables defined by
	\begin{equation}\label{eq: Korsakova}
		\begin{split}
		 f(z,z_{n+1},z_{n+2}) =& g(z) + \sum_{i=1}^n \alpha_i z_iz_{n+1}+z_{n+1}z_{n+2}, \\
		 f'(z,z_{n+1},z_{n+2}) =& g(z) + \sum_{i=1}^n \alpha_i z_i(z_{n+1}+z_{n+2})\\+&z_{n+1}z_{n+2},
		\end{split}
	\end{equation}
	for all $z \in \F_2^n$ and $z_{n+1},z_{n+2} \in \F_2$, are bent functions for all $\alpha_1, \ldots , \alpha_n \in \F_2$.
\end{theo}

Nevertheless, these methods fall under the concatenation framework given by $f=g||h||g||(h+1)$ and the class inclusion of these functions is then easily determined.

\begin{prop}\label{prop:usingkorsakova}
	Let $g$ be a bent function in $n$ variables. Let $f$ and $f'$ be the bent functions in $(n+2)$-variables defined in Eq.~\eqref{eq: Korsakova}. Then, the functions $f$ and $f'$ are in the $\cM^{\#}$ class if and only if the function $g$ is in the $\cM^{\#}$ class.
\end{prop}
\begin{proof}
	Note that $f$ and $f'$ are extended affine equivalent, hence it is enough to investigate the class inclusion for one of them; therefore, we will prove the result for $f$. By looking at $f(z,0,0)$, $f(z,1,0)$, $f(z,0,1)$ and $f(z,1,1)$, we see that $f=f_1 \vert \vert f_2 \vert \vert f_3 \vert \vert f_4$, where $f_1(z)=f_3(z)=g(z)$, $f_2(z)=f_4(z)+1=g(z)+ \sum_{i=1}^{n} \alpha_iz_i$. Since the functions $g(z)$ and $g(z)+ \sum_{i=1}^{n} \alpha_iz_i$ have the same $\cM$-subspaces, the result follows from Theorem \ref{th: insideMMgh}.
\end{proof}
Consequently, we provide an alternative proof of existence of cubic bents functions outside $\mathcal{M}^\#$ on $\F_2^n$ for all $n\ge 10$.
	\begin{cor}
		Cubic bent functions on $\F_2^n$ outside $\mathcal{M}^\#$ exist for all $n\ge 10$.
	\end{cor}
	\begin{proof}
		In Theorem~\ref{th: Korsakova}, take $g\in\mathcal{B}_{10}$ as $h^{10}_3$ or $h^{10}_4$ from~\cite[Table 4]{Polujan2020}, which are both outside $\mathcal{M}^\#$.
\end{proof}
\subsection{Applying suitable affine transforms}
The class inclusion properties are substantially affected by applying suitable affine transformations to bent functions used in 4-bent concatenation.
\begin{theo}\label{theo: affinetransform} 
	Let $g \in \mathcal{B}_n$ be a bent function, $n\geq 6$, in the $\cM^{\#}$ class, and let $q \in \mathcal{B}_n$ be a bent function with a unique $n/2$-dimensional $\cM$-subspace. Then, there exist two linear permutations $A$ and $B$ of $\F_2^n$ such that for $h=q \circ A$ and $h'= q \circ B$ in  $\mathcal{B}_n$, the function $f \in \mathcal{B}_{n+2}$ defined by
	$
	f=g \vert \vert h \vert \vert g \vert \vert (h+1)
	$
	is a bent function inside the $\cM^{\#}$ class, and the function $f' \in \mathcal{B}_{n+2}$ defined by
	$
	f'=g \vert \vert h' \vert \vert g \vert \vert (h'+1)
	$
	is a bent function outside the $\cM^{\#}$ class. 
\end{theo}
\begin{proof}
	Let $a,b \in \F_2^n$ be two elements such that $D_aD_bg \neq 0$ (we know such two elements exist, otherwise $g$ would be affine), and let $V$ be the unique $(n/2)$-dimensional $\cM$-subspace of $q$.
	Let $B$ be any linear isomorphism such that $ \{a,b\} \subset B^{-1}(V)$. The subspace $B^{-1}(V)$ is the unique $(n/2)$-dimensional $\cM$-subspace of $q \circ B$. Since $\{a,b\} \subset B^{-1}(V)$ and $D_aD_bg \neq 0$, we deduce that $g$ and $q \circ B$ do not share an $(n/2)$-dimensional $\cM$-subspace. Setting $h'=q \circ B$, from Theorem \ref{th: insideMMgh}, we deduce that 
	$	f'=g \vert \vert h' \vert \vert g \vert \vert (h'+1) \not \in \cM^{\#}$.
	Notice that $h'$ also admits a unique $(n/2)$-dimensional $\cM$-subspace as $h$ does.
	
	On the other hand, since $g \in \cM^{\#}$, it has at least one $(n/2)$-dimensional $\cM$-subspace, denote it by $W$. Let $A$ be any linear isomorphism such that $A(W)=V$, and set $h=q \circ A$. Then, $W$ is an $(n/2)$-dimensional $\cM$-subspace of both $g$ and $h$. By Theorem \ref{th: insideMMgh}, we have that 
	$	f=g \vert \vert h \vert \vert g \vert \vert (h+1) \in \cM^{\#}$.
\end{proof}

\section*{Acknowledgment}
The authors would like to thank the anonymous reviewers for their valuable comments, which helped to improve the presentation of the results.


\begin{thebibliography}{10}
	\providecommand{\url}[1]{#1}
	\csname url@samestyle\endcsname
	\providecommand{\newblock}{\relax}
	\providecommand{\bibinfo}[2]{#2}
	\providecommand{\BIBentrySTDinterwordspacing}{\spaceskip=0pt\relax}
	\providecommand{\BIBentryALTinterwordstretchfactor}{4}
	\providecommand{\BIBentryALTinterwordspacing}{\spaceskip=\fontdimen2\font plus
		\BIBentryALTinterwordstretchfactor\fontdimen3\font minus
		\fontdimen4\font\relax}
	\providecommand{\BIBforeignlanguage}[2]{{%
			\expandafter\ifx\csname l@#1\endcsname\relax
			\typeout{** WARNING: IEEEtranS.bst: No hyphenation pattern has been}%
			\typeout{** loaded for the language `#1'. Using the pattern for}%
			\typeout{** the default language instead.}%
			\else
			\language=\csname l@#1\endcsname
			\fi
			#2}}
	\providecommand{\BIBdecl}{\relax}
	\BIBdecl
	
	\bibitem{Decom}
	\BIBentryALTinterwordspacing
	A.~Canteaut and P.~Charpin, ``\href{https://doi.org/10.1109/TIT.2003.814476}{Decomposing bent functions},'' \emph{IEEE
		Transactions on Information Theory}, vol.~49, no.~8, pp. 2004--2019, 2003.
	\BIBentrySTDinterwordspacing
	
	\bibitem{CarlMes2016}
	\BIBentryALTinterwordspacing
	C.~Carlet and S.~Mesnager, ``\href{https://doi.org/10.1007/s10623-015-0145-8}{Four decades of research on bent functions},''
	\emph{Designs, Codes and Cryptography}, vol.~78, no.~1, pp. 5--50, Jan 2016.
	\BIBentrySTDinterwordspacing
	
	\bibitem{Dillon}
	\BIBentryALTinterwordspacing
	J.~F. Dillon, ``\href{https://doi.org/10.13016/M2MS3K194}{Elementary {H}adamard difference sets},'' Ph.D. dissertation,
	University of Maryland, 1974.
	\BIBentrySTDinterwordspacing
	
	\bibitem{SHCF}
	\BIBentryALTinterwordspacing
	S.~Hod{\v{z}}i{\'{c}}, E.~Pasalic, and Y.~Wei, ``\href{https://doi.org/10.1007/s10623-020-00760-9}{A general framework for
	secondary constructions of bent and plateaued functions},'' \emph{Designs,
		Codes and Cryptography}, vol.~88, no.~10, pp. 2007--2035, Oct 2020.
	\BIBentrySTDinterwordspacing
	
	\bibitem{Korsakova2013}
	\BIBentryALTinterwordspacing
	E.~P. Korsakova, ``\href{https://www.mathnet.ru/php/archive.phtml?wshow=paper&jrnid=da&paperid=745&option_lang=eng}{Graph classification for quadratic bent functions in $6$
	variables},'' \emph{Diskretn. Anal. Issled. Oper.}, vol.~20, no.~5, pp.
	45--57, 2013.
	\BIBentrySTDinterwordspacing
	
	\bibitem{MM73}
	\BIBentryALTinterwordspacing
	R.~L. McFarland, ``\href{https://doi.org/10.1016/0097-3165(73)90031-9}{A family of difference sets in non-cyclic groups},''
	\emph{Journal of Combinatorial Theory, Series A}, vol.~15, no.~1, pp. 1--10,
	1973.
	\BIBentrySTDinterwordspacing
	
	\bibitem{Mesnager}
	\BIBentryALTinterwordspacing
	S.~Mesnager, \href{https://doi.org/10.1007/978-3-319-32595-8}{\emph{{B}ent {F}unctions: {F}undamentals and {R}esults}},
	1st~ed.\hskip 1em plus 0.5em minus 0.4em\relax Springer Cham, 2016.
	\BIBentrySTDinterwordspacing
	
	\bibitem{Bent_Decomp2022}
	\BIBentryALTinterwordspacing
	E.~Pasalic, A.~Bapi{\'{c}}, F.~Zhang, and Y.~Wei, ``\href{https://doi.org/10.1007/s10623-023-01204-w}{Explicit infinite families
	of bent functions outside the completed {M}aiorana--{M}c{F}arland class},''
	\emph{Designs, Codes and Cryptography}, vol.~91, no.~7, pp. 2365--2393, Jul
	2023.
	\BIBentrySTDinterwordspacing
	
	\bibitem{PPKZ2023}
	\BIBentryALTinterwordspacing
	E.~Pasalic, A.~Polujan, S.~Kudin, and F.~Zhang, ``	\href{https://doi.org/10.1109/TIT.2024.3352824}{Design and analysis of bent
	functions using {$\mathcal{M}$}-subspaces},'' \emph{To appear in IEEE
		Transactions on Information Theory}, pp. 1--1, 2024.
	\BIBentrySTDinterwordspacing
	
	\bibitem{PPKZ_BFA23_CCDS}
	\BIBentryALTinterwordspacing
	A.~Polujan, E.~Pasalic, S.~Kudin, and F.~Zhang, ``\href{https://doi.org/10.48550/arXiv.2310.10162}{Bent functions satisfying the dual bent condition and permutations with the {$(\mathcal{A}_m)$} property},''
	\emph{Submitted}, 2023.
	\BIBentrySTDinterwordspacing
	
	\bibitem{Polujan2020}
	\BIBentryALTinterwordspacing
	A.~Polujan and A.~Pott, ``\href{https://doi.org/10.1007/s10623-019-00712-y}{Cubic bent functions outside the completed
	{M}aiorana-{M}c{F}arland class},'' \emph{Designs, Codes and Cryptography},
	vol.~88, no.~9, pp. 1701--1722, Sep 2020.
	\BIBentrySTDinterwordspacing
	
	\bibitem{Tokareva}
	N.~Tokareva, \emph{{Bent Functions. Results and Applications to Cryptography}},
	1st~ed.\hskip 1em plus 0.5em minus 0.4em\relax Academic Press, 2015.
	
	\bibitem{ZPBBInfComp}
	\BIBentryALTinterwordspacing
	F.~Zhang, E.~Pasalic, A.~Bapić, and B.~Wang, ``\href{https://doi.org/10.1016/j.ic.2024.105149}{Constructions of several special classes of cubic bent functions outside the completed
	{M}aiorana-{M}c{F}arland class},'' \emph{Information and Computation}, p.
	105149, 2024. 
	\BIBentrySTDinterwordspacing
	
	\bibitem{Zheng2001}
	\BIBentryALTinterwordspacing
	Y.~Zheng and X.-M. Zhang, ``\href{https://doi.org/10.1109/18.915690}{On plateaued functions},'' \emph{IEEE Transactions
		on Information Theory}, vol.~47, no.~3, pp. 1215--1223, 2001.
	\BIBentrySTDinterwordspacing
	
\end{thebibliography}
\end{document}